\documentclass[pdflatex,sn-mathphys-num]{sn-jnl}

\usepackage{graphicx}%
\usepackage{multirow}%
\usepackage{amsmath,amssymb,amsfonts}%
\usepackage{amsthm}%
\usepackage{mathrsfs}%
\usepackage[title]{appendix}%
\usepackage{xcolor}%
\usepackage{textcomp}%
\usepackage{manyfoot}%
\usepackage{booktabs}%
\usepackage{algorithm}%
\usepackage{algorithmicx}%
\usepackage{algpseudocode}%
\usepackage{listings}%
\usepackage{algorithm}
\usepackage{algpseudocode}

\theoremstyle{thmstyleone}%
\newtheorem{theorem}{Theorem}
\newtheorem{proposition}[theorem]{Proposition}%

\theoremstyle{thmstyletwo}%
\newtheorem{remark}{Remark}%

\theoremstyle{thmstylethree}%

\newcommand{\LL}{\mathbb{L}}
\begin{document}

\title[Article Title]{A Time-Varying Branching Process Approach to Model Self-Renewing Cells}

\author[1]{\fnm{Huyen} \sur{Nguyen}}\email{huyen.d.nguyen@uconn.edu}

\author[1]{\fnm{Haim} \sur{Bar}}\email{haim.bar@uconn.edu}

\author[1]{\fnm{Zhiyi} \sur{Chi}}\email{zhiyi.chi@uconn.edu}

\author[1]{\fnm{Vladimir} \sur{Pozdnyakov}}\email{vladimir.pozdnyakov@uconn.edu}

\affil[1]{\orgdiv{Department of Statistics}, \orgname{University of Connecticut}, \orgaddress{\street{215 Glenbrook Road}, \city{Storrs}, \postcode{06269}, \state{Connecticut}, \country{United States}}}




\abstract{Stem cells, through their ability to produce daughter stem cells and differentiate into specialized cells, are essential in the growth, maintenance, and repair of biological tissues. Understanding the dynamics of cell populations in the proliferation process not only uncovers proliferative properties of stem cells, but also offers insight into tissue development under both normal conditions and pathological disruption. In this paper, we develop a continuous time branching process model with time-dependent offspring distribution to characterize stem cell proliferation process. We derive analytical expressions for mean, variance, and autocovariance of the stem cell counts, and develop likelihood-based inference procedures to estimate model parameters. Particularly, we construct a forward algorithm likelihood to handle situations when some cell types cannot be directly observed. Simulation results demonstrate that our estimation method recovers the time-dependent division probabilities with good accuracy.}




\maketitle

\section{Introduction}\label{sec1}

Cell proliferation is a fundamental process underlying the growth, maintenance, and repair of biological tissues. Stem cells, in particular, possess the ability to produce identical daughter stem cells and differentiate into specialized cell types that sustain tissue function. There are three types of stem cell division and differentiation outcomes: symmetrical self-renewal to generate two stem cells, asymmetrical division to generate one stem cell and one differentiated cell, and symmetrical differentiation to generate two differentiated cells. The ability of stem cells to both self-renew and differentiate plays a pivotal role both in generative processes during early development and in tissue regeneration following injury. For example, during embryonic development, stem cells drive tissue formation through programmed proliferation and differentiation, as seen in processes such as ependymogenesis and neurogenesis \cite{Co2018, Ba2018}. Hematopoietic stem cells give rise to all blood cell types - red blood cells, white blood cells, and platelets - to maintain lifelong blood production and regenerate the blood system after injury \cite{marciniak2009modeling, wilson2008hematopoietic}. In oncology, cancer stem cells also display self-renewal and differentiation ability to drive and sustain tumor growth. Distinct differences in stemness and proliferative potential between cancer stem cells and non-stem cancer cells have been observed across various tumor types such as brain, breast, colon, and prostate cancer \cite{al2003prospective, cammareri2008isolation, fioriti2008cancer, singh2003identification}.

It is important to study the dynamics of cell populations in cell proliferation process since these dynamics connect individual cell behaviors to tissue-level outcomes, including how tissues grow, maintain homeostasis, and regenerate after injury. In the aforementioned ependymogenesis process, regional stem cells within the ventricular and subventricular zones proliferate to produce both stem cells and differentiated ependymal cells that line the brain’s ventricles, thereby creating a critical barrier between the cerebrospinal fluid and the interstitial fluid of the surrounding brain parenchyma \cite{Co2018}. In the context of cancer study, these population dynamics are particularly important since they reflect the growth of tumor and disease progression.

Quantitative modeling of these processes can provide a framework to predict population behavior and discover intrinsic proliferative properties of stem cell self-renewal and differentiation ability. Deterministic models with ordinary differential equations are utilized to describe the temporal dynamics of cell population on average \cite{marciniak2009modeling, beretta2012mathematical, weekes2014multicompartment, chen2016overshoot}. However, since the dynamics of cell populations emerge from the stochastic behavior of individual stem cells, each of which may divide or differentiate at random times, probabilistic modeling approaches, including birth-death process, Gaussian jump model, and non-homogeneous Markov chain \cite{Tu2009, Ba2018, Sc2022}, are more verisimilar for describing cell proliferation process. Branching process provides a natural and mathematically rigorous framework for describing stochastic reproduction phenomena. In the classical Galton-Watson process and its extensions, each individual acts independently and give rise to a random number of offspring according to specified probability distributions \cite{watson1875probability,harris1948branching, harris1963theory, asmussen1983branching}. The branching process models have become an important tool to investigate population size, growth and extinction in biology, epidemiology, and genetics. In the context of cell proliferation, branching process models have been applied to describe stem cell lineages, where each stem cell may be viewed as an independent reproducing unit. MacMillian et al. \cite{macmillan2011seeing} and Miguez \cite{miguez2015branching} both developed models under the branching process theory for the stem cell proliferation and differentiation during embryonic development. In these models, time is discretized using the average cell cycle length as the unit of generations. Instead of tracking individual cell divisions in continuous time, the model advances steps corresponding to the mean duration of a cell cycle. In addition, MacMillian et al. \cite{macmillan2011seeing} explored the multi-type branching process extension to simulate data for a cell proliferation with more than one cell type with proliferative potential.

However, the classical branching process models typically assume that the reproduction parameters are constant through time. This assumption simplifies the analysis but may not be able to capture the full stochastic nature of cell proliferation processes. Experimental studies have revealed that the probabilities governing self-renewal and differentiation are not always stationary. For example, data collected by Scanlon et al. \cite{Sc2022} on bipotent megakaryocytic erythroid progenitor (MEP) clonal
expansion and differentiation shows that probabilities of division outcomes change over time: self-renewing events are more likely earlier in the process and differentiating events are more likely toward the end of the process. The previously discussed branching process framework in Miguez \cite{miguez2015branching} adjusted to the time-dependent proliferation by obtaining the average probabilities of division outcomes at different developmental stages.

In this paper, we develop a continuous-time, time-dependent branching process for stem cell proliferation to address the limitation of the constant reproduction branching process model. In our formulation, the probabilities of division and differentiation outcomes are functions of time. To reflect the variability inherent in stem cell populations, our model incorporates three distinct cell types: viable stem cells capable of division, nonviable stem cells that do not have proliferative potential, and terminally differentiated cells. We also incorporate stochasticity in the timing of the cell division events by assuming that each division occurs at a random time. In particular, the
interarrival times between stem cell divisions are exponentially distributed with rate proportional to current number of viable stem cells. This reflects the assumption that a greater number of viable stem cells leads to more frequent stem cell divisions and shorter waiting times in-between. Our work makes the following contributions. We first establish a rigorous mathematical formulation of the branching process model and derive analytical expressions for key parameters, including mean population size, variance, and lower bound of the expected extinction time for stem cells population. We also construct likelihood function based on the observed counts of each cell type for the probabilities of division and differentiation outcomes, providing a framework for statistical inference. Finally, we develop a forward-algorithm-based likelihood approach to situations where viable and nonviable stem cells cannot be distinguished, allowing for estimation of model parameters even when some cell types cannot be directly observed. This gives a practical method to model stem cell dynamics in experiments where direct observation of all cell types is not possible.

The remainder of this paper is organized as follows. In Section \ref{model}, we introduce the time-dependent branching process model and derive analytical expressions for key quantities. Section \ref{mle} presents the likelihood-based framework for parameter estimation, including the forward-algorithm approach for partially observed cell types. Section \ref{simulation} provides a simulation study to assess the performance of the likelihood-based estimator, Section \ref{application} illustrates an example of experiment data application, and Section \ref{discussion} concludes with a discussion of the findings and potential directions for future work.

\section{Time-Dependent Branching Process Model for Stem Cell Proliferation Process}\label{model}
\subsection{Model Description}
We consider a stochastic cell proliferation process consisting of populations of stem cells and differentiated cells (FC). Stem cells are classified as viable stem cells (SC) and nonviable (or duds) stem cells (DC). Let $X(t), Y(t),$ and $Z(t)$ denote the number of viable stem cells, differentiated cells, and nonviable stem cells at time $t$, respectively. Only viable stem cells are capable of division and differentiation. The stochastic nature of the process arises from the random timing and outcome of division events. Each viable stem cell divides independently in continuous time with rate $r$. Biologically, this rate could be interpreted as the average frequency that a viable stem cell exits quiescence and enters the cell cycle. Consequently, when there are $X(t)$ viable stem cells at time $t$, the waiting time until the next division event in the entire population is exponentially distributed with rate $rX(t)$. Thus, the overall frequency of division events increases with the size of the viable stem cell population.

At time $t$, a viable stem cell can undergo one of the four types of division with the time-dependent probabilities $p_1(t), p_2(t), p_3(t)$, and $p_4(t)$:
\begin{align*}
    &\text{Symmetric self-renewal} &  SC \rightarrow SC + SC,  &  \quad \text{probability }p_1(t)  \\
    &\text{Asymmetric division} & SC \rightarrow SC + FC, & \quad \text{probability }p_2(t) \\
    &\text{Symmetric differentiation} &  SC \rightarrow FC + FC, & \quad \text{probability }p_3(t)\\
    &\text{Self renewal with nonviable stem cell} & SC \rightarrow SC + DC,& \quad \text{probability }p_4(t)
\end{align*}
The division probabilities $p_j(t)$ are time-dependent. For example, $p_1(t)$ may dominate early to expand the stem cell pool, while $p_3(t)$ may increase toward the end of the cell proliferation process. Since $p_j(t)$ are probabilities, they need to satisfy the conditions $0 \leq p_j(t) \leq 1$ and $\sum_{j=1}^n p_j(t) = 1 \ \forall t$. To illustrate, we assume the following form for $p_j(t)$
\begin{equation} \label{lorentzian}
    \begin{cases}
        p_1(t) = \frac{p_1}{1+c_1(t-m_1)^2},\\
        p_2(t) = \frac{p_2}{1+c_2(t-m_2)^2},\\
        p_3(t) = 1- \frac{p_1}{1+c_1(t-m_1)^2} - \frac{p_2}{1+c_2(t-m_2)^2}- \frac{p_4}{1+c_4(t-m_4)^2}, \\
        p_4(t) = \frac{p_4}{1+c_4(t-m_4)^2},
    \end{cases}
\end{equation}
with the constraints, $0 \leq p_1, p_2, p_4 \leq 1$, $p_1 + p_2+p_4 \leq 1$, and $c_1, c_2, c_4 \geq 0$. In this example, the functions $p_1(t), p_2(t),$ and $p_4(t)$ are Lorentzian shape commonly used in physics. In the context of probability and statistics, they can be rewritten as a scaled version of the Cauchy probability density function by reparameterizing $c_i$. The functions are parameterized by $p_i$, which are the maximum of the function attained at $t = m_i$, and $c_i$, which governs the rate of decay away from the maximum. With this form of functions for $p_j(t)$, $p_3(t) \rightarrow 1$ as $t\rightarrow \infty$. As a result, the differentiation events will dominate toward the end of the proliferation process. This behavior aligns with the dynamics observed in the MEP clonal expansion and differentiation data observed by Scanlon et al. \cite{Sc2022}.

This proliferation process can be naturally represented as a branching process, in which each viable stem cell acts as a branching individual producing “offspring” according to the division probabilities $p_j(t)$. Differentiated cells correspond to terminal offspring, and nonviable stem cells correspond to individuals that cannot further proliferate. When the number of viable stem cell $X_i$ is positive, a stem cell division will occur at time $T_i$, with interarrival times that are exponentially distributed at a rate proportional to $X_i$. Thus, given division times $T_i$, the change of viable stem cells, differentiated cells, and nonviable stem cells have the following marginal distributions
\begin{equation}
    P(\Delta X_i, \Delta Y_i, \Delta Z_i \vert T_i) = \begin{cases}
        p_1(T_i), &  (\Delta X_i, \Delta Y_i, \Delta Z_i) = (1,0,0)\\
        p_2(T_i), & (\Delta X_i, \Delta Y_i, \Delta Z_i) = (0, 1, 0)\\
        p_3(T_i), & (\Delta X_i, \Delta Y_i, \Delta Z_i) = (-1, 2, 0)\\
        p_4(T_i), & (\Delta X_i, \Delta Y_i, \Delta Z_i) = (0, 0, 1).
    \end{cases}
    \label{marginal}
\end{equation}
A deterministic mean-field model can capture the average behavior of the cell population by assuming individual viable stem cells reproduce at their average rate. Let $S(t), F(t), D(t)$ denote the deterministic population trajectories of viable stem cells, differentiated cells, and nonviable stem cells. If we assume that the average rates at which viable stem cells, differentiated cells, and nonviable stem cells are reproduced are $r[p_1(t) - p_3(t)]S(t)$, $r[p_2(t) + 2p_3(t)] S(t)$, and $rp_4(t)S(t)$, respectively, we obtain the ordinary differential equation
\begin{equation}\label{transition1}
\begin{pmatrix}
\frac{dS(t)}{dt}\\\frac{dD(t)}{dt}\\\frac{dF(t))}{dt}
\end{pmatrix} =
\begin{pmatrix}
r[p_1(t)-p_3(t)]S(t)\\rp_4(t)S(t)\\r[p_2(t)+2p_3(t)]S(t)
\end{pmatrix} \,.
\end{equation}
A limitation of the deterministic model is that it treats the probabilities of division and differentiation events as average rates, and therefore does not capture the randomness of the underlying process. As a result, the deterministic model cannot account for the variability between realizations or extinction events.
\subsection{Mean and Covariance Structure of Viable Stem Cell Population}
In this subsection, we derive the expected value of the viable stem cell population in the branching process model and show that it coincides with the solution of the deterministic differential equation. We also compute the variance and autocovariance, which measure the variability of the population at a single time point and across time points, respectively.
\begin{proposition}[Moments of Number of Stem Cells] \label{prop:moments}
For each fixed time point, the expected value, variance, and autocovariance of the stem cell population are given by:

\begin{enumerate}
    \item \textbf{Expected Value (Mean Function).}
    \[
     \mathbb{E}[X(t)] = S(t) = S_0 \exp \Big\{r \int_0^t [ p_1(u) - p_3(u) du]\Big\}.
    \]

    \item \textbf{Variance.}
    \[
        \operatorname{Var}(X(t))
        = r S(t)^2 \int_0^t \frac{p_1(u) + p_3(u)}{S(u)}du.
    \]

    \item \textbf{Autocovariance Function.}
    For $t>u,$
    \[
       \operatorname{Cov}(X(t), X(u)) = S(t)S(u) r\int_0^u\frac{p_1(v) + p_3(v)}{S(v)}dv = \frac{S(t)}{S(u)} \operatorname{Var}(u).
    \]
\end{enumerate}
\end{proposition}

\begin{proof}
The solution of the differential equation \eqref{transition1} with initial condition
$S(0)=S_0$ is given by
\begin{equation}
\label{eq:S_solution}
S(t)=S_0 \exp\!\left\{ r\int_0^t \bigl[p_1(u)-p_3(u)\bigr]\,du \right\}.
\end{equation}
We first show that $E[X(t)] = S(t)$. Note that for a small $\Delta>0$, conditional on $X(t)$, the arrival times in the interval $[t, t+\Delta]$ essentially follow a Poisson process with rate $rX(t)$. Since $p_i(t)$ are continuous functions, we have
\begin{equation}
\label{eq:increment_mean}
\lim_{\Delta\to 0}\frac{1}{\Delta}E\!\left[X(t+\Delta)-X(t)\mid X(t)\right]=\bigl[p_1(t)-p_3(t)\bigr]\,rX(t),
\end{equation}
with probability 1.
At each arrival time, the number of stem cells can increase by at most one. As a consequence, the stochastic process $X(t)$ is dominated by the Yule process. Therefore, by the dominated convergence theorem  (see Siegrist (2021, Eq. (16.23.40) \cite{siegrist2021})) we obtain
\begin{equation}
\frac{d}{dt}E[X(t)]
=\bigl[p_1(t)-p_3(t)\bigr]\,rE[X(t)].
\end{equation}
Since $E[X(0)]=S_0$, uniqueness of solutions to linear ODEs implies
$E[X(t)]=S(t)$.

Next, we derive the variance of $X(t)$. Define $M(t)=E[X(t)^2]$ and $V(t)=\operatorname{Var}(X(t))=M(t)-S(t)^2$.
In a similar fashion, conditioning on $X(t)$ gives us
\begin{equation}
\label{eq:increment_second}
\lim_{\Delta\to 0}\frac{1}{\Delta}E\!\left[X(t+\Delta)^2-X(t)^2\mid X(t)\right]=\bigl[2X(t)(p_1(t)-p_3(t))+(p_1(t)+p_3(t))\bigr]\,rX(t)
\end{equation}
with probability 1. Taking expectations and letting $\Delta\to 0$ yields
\begin{equation}
\label{eq:Mprime}
M'(t)
=2\bigl[p_1(t)-p_3(t)\bigr]\,rM(t)
+\bigl[p_1(t)+p_3(t)\bigr]\,rS(t).
\end{equation}
Since $S'(t)=\bigl[p_1(t)-p_3(t)\bigr]rS(t)$, differentiation of
$V(t)=M(t)-S(t)^2$ yields
\begin{equation}
\label{eq:Vprime}
V'(t)
=2\bigl[p_1(t)-p_3(t)\bigr]\,rV(t)
+\bigl[p_1(t)+p_3(t)\bigr]\,rS(t).
\end{equation}
Assuming $V(0)=0$, the solution of \eqref{eq:Vprime} is
\begin{equation}
\label{eq:Vsolution}
V(t)
=rS(t)^2\int_0^t \frac{p_1(u)+p_3(u)}{S(u)}\,du,
\end{equation}
where $P(t)=\int_0^t[p_1(u)-p_3(u)]\,du$. Finally, for $t>u$, the tower property and conditional expectation give
\begin{equation}
\begin{split}
E[X(t)X(u)]
&=E\!\left[E[X(t)\mid X(u)]\,X(u)\right] \\
&=E\!\left[X(u)^2 \exp\{r[P(t)-P(u)]\}\right] \\
&=\frac{S(t)}{S(u)}\,E[X(u)^2].
\end{split}
\end{equation}
Using \eqref{eq:Vsolution} with $M(u)=S(u)^2+V(u)$, we obtain
\begin{equation}
\begin{split}
E[X(t)X(u)]-S(t)S(u)
&=\frac{S(t)}{S(u)}\,V(u).
\end{split}
\end{equation}
This completes the proof.
\end{proof}
We validate the theoretical moments of stem cell count using simulated data. Details on how the data is simulated is discussed in section \ref{simulation}. Figure \ref{fig:sim_plot} (left) and table \ref{tab:autocorrelation} illustrate an example of empirical mean, variance and autocorrelation of stem cell counts obtained from simulated data. The empirical results closely match with their theoretical counterparts derived in proposition \ref{prop:moments}.
\begin{figure}[ht]
\centering

\begin{minipage}{0.48\textwidth}
\centering
\includegraphics[width=\linewidth]{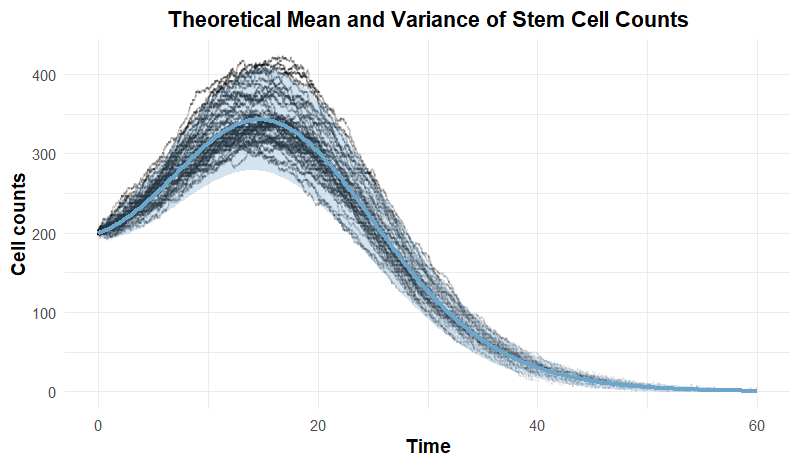}
\end{minipage}
\hfill
\begin{minipage}{0.48\textwidth}
\centering
\includegraphics[width=\linewidth]{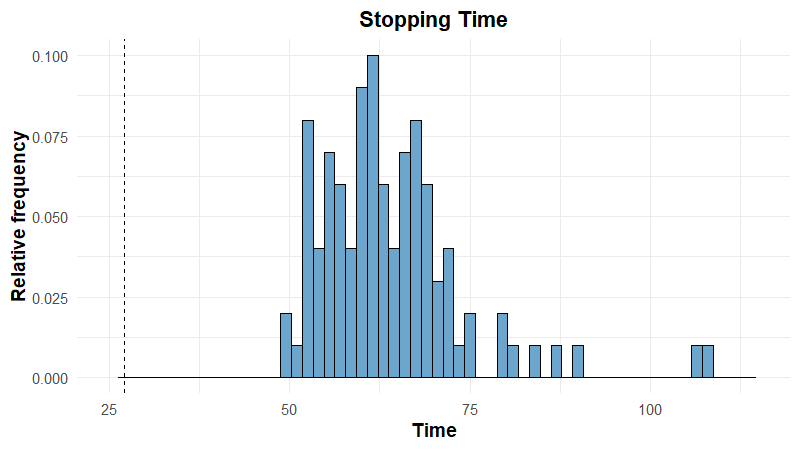}
\end{minipage}

\caption{Left: Fifty realizations of simulated stem cell count trajectories with initial count of 200, division rate of 0.2, and division probabilities $p_1(t) = \frac{0.55}{1 + 0.005(t-4)^2}$, $p_2(t) = \frac{0.15}{1 + (t-12)^2}$, $p_4(t) = \frac{0.20}{1+(t-20)^2}$, and $p_3(t) = 1-p_1(t)-p_2(t)-p_4(t)$. The blue curve shows the expected cell count and the shaded band indicates $\pm$2 standard deviations. Right: Distribution of stopping time, i.e. the time when all the viable stem cells are exhausted from 100 realizations of the simulated stem cell count trajectories with the same parameters, with the minimum expected stopping times is $\tau_{\min} = 26.491$.}
\label{fig:sim_plot}
\end{figure}

\begin{table}[ht]
\centering

\begin{minipage}{0.48\textwidth}
\centering
\textbf{Theoretical autocorrelation}

\vspace{0.3em}
\footnotesize
\begin{tabular}{c ccccc}
\toprule
$t$ & 5 & 10 & 15 & 20 & 25 \\
\midrule
 5  & 1.000 & 0.763 & 0.666 & 0.598 & 0.533 \\
10  & 0.763 & 1.000 & 0.872 & 0.783 & 0.698 \\
15  & 0.666 & 0.872 & 1.000 & 0.899 & 0.801 \\
20  & 0.598 & 0.783 & 0.899 & 1.000 & 0.891 \\
25  & 0.533 & 0.698 & 0.801 & 0.891 & 1.000 \\
\bottomrule
\end{tabular}
\end{minipage}
\hfill
\begin{minipage}{0.48\textwidth}
\centering
\textbf{Empirical autocorrelation}

\vspace{0.3em}
\footnotesize
\begin{tabular}{c ccccc}
\toprule
$t$ & 5 & 10 & 15 & 20 & 25 \\
\midrule
 5  & 1.000 & 0.742 & 0.588 & 0.552 & 0.535 \\
10  & 0.742 & 1.000 & 0.884 & 0.761 & 0.758 \\
15  & 0.588 & 0.884 & 1.000 & 0.877 & 0.849 \\
20  & 0.552 & 0.761 & 0.877 & 1.000 & 0.891 \\
25  & 0.535 & 0.758 & 0.849 & 0.891 & 1.000 \\
\bottomrule
\end{tabular}
\end{minipage}

\caption{Theoretical versus empirical autocorrelation of stem cell counts at time points
$t = 5, 10, 15, 20, 25$. The empirical autocorrelation is calculated from realizations of
simulated stem cell count trajectories with initial count of 200, division rate of 0.2,
and division probabilities $p_1(t) = \frac{0.55}{1 + 0.005(t-4)^2}$,
$p_2(t) = \frac{0.15}{1 + (t-12)^2}$,
$p_4(t) = \frac{0.20}{1+(t-20)^2}$,
and $p_3(t) = 1-p_1(t)-p_2(t)-p_4(t)$.}
\label{tab:autocorrelation}
\end{table}

\begin{remark}
Suppose that
\[
P(t)=\int_0^t \bigl[p_1(u)-p_3(u)\bigr]\,du \longrightarrow -\infty
\qquad \text{as } t\to\infty.
\]
Then both the mean and the variance of $X(t)$ converge to zero.

The convergence $E[X(t)]\to0$ follows immediately from the explicit expression
\[
E[X(t)] = S_0 \exp\!\left\{ rP(t) \right\}.
\]
To analyze the variance, recall that
\[
V(t)
= r S_0 \exp\!\left\{ 2rP(t) \right\}
\int_0^t \exp\!\left\{ -rP(u) \right\}
\bigl[p_1(u)+p_3(u)\bigr]\,du.
\]
Since $0 \le p_1(t)+p_3(t)\le1$, we obtain the bound
\begin{equation}
\label{eq:variance_bound}
V(t)
\le r S_0 \exp\!\left\{ 2rP(t) \right\}
\int_0^t \exp\!\left\{ -rP(u) \right\}\,du.
\end{equation}
Applying l’Hôpital’s rule to the ratio
\[
\frac{\int_0^t \exp\!\left\{ -rP(u) \right\}\,du}
{\exp\!\left\{ -2rP(t) \right\}},
\]
we find
\[
\lim_{t\to\infty}
\frac{\int_0^t \exp\!\left\{ -rP(u) \right\}\,du}
{\exp\!\left\{ -2rP(t) \right\}}
=
\lim_{t\to\infty}
\frac{\exp\!\left\{ rP(t) \right\}}
{-2r\bigl[p_1(t)-p_3(t)\bigr]}.
\]
Under the assumption $P(t)\to-\infty$, the numerator converges to zero, while the denominator remains bounded away from zero whenever $p_1(t)-p_3(t)$ does not vanish identically. Consequently, the right-hand side of \eqref{eq:variance_bound} converges to zero, and hence
\[
\lim_{t\to\infty} V(t)=0.
\]
\end{remark}
Intuitively, this result is expected because over time, as differentiation events dominate self-renewal events, stem cells are lost faster than they are replenished and become extinct. Moreover, when the initial number of stem cells is large, for any fixed time point $t$ the stochastic process is close to their expected values. Specifically, we have the following limit result.
\begin{proposition}
For fixed $t$, as $S_0 \rightarrow \infty$,
\begin{equation}
    \frac{X(t) - S(t)}{S(t)} \xrightarrow{L_2} 0.
\end{equation}
\end{proposition}

\begin{proof} Since we have $E[X(t)] = S(t)$, we only need to show $\operatorname{Var}\Big( \frac{X(t) - S(t)}{S(t)} \Big) \longrightarrow 0$ as $ S_0 \longrightarrow \infty$.
\begin{equation}
    \operatorname{Var} \Big( \frac{X(t) - S(t)}{S(t)} \Big) = \frac{r}{S_0} \int_0^t \frac{p_1(u) + p_3(u)}{\exp \{ rP(u)\}} du \longrightarrow 0.
\end{equation}
\end{proof}
This result implies that, when starting with a large number of initial stem cells, the stem-cell population in the branching process can be well approximated by the deterministic model with the differential equation in (\ref{transition1}).
\subsection{Stopping Time}
Since viable stem cells are the only cell type with proliferative capability in our model, the process terminates when they become extinct. In this section, we analyze the stopping time of the process, defined as the time at which the number of viable stem cells reach 0.

We first consider the lower bound of the stopping times. Under the model assumption that interarrival times between two consecutive division events follows an exponential distribution with rate proportional to the current viable stem cells, the average interarrival time between two consecutive division events $T_{i+1} - T_i$ is $1/(rS_i)$. The minimum stopping times occurs when every event is a differentiation event, in which no new viable stem cell is created and the number of viable stem cells decreases by 1. Then the minimum expected stopping time is
$$\tau_{\min} = \frac{1}{rS_0} + \frac{1}{r (S_0 -1)} + \cdots + \frac{1}{r}.$$
Using the harmonic sum approximation,
\begin{equation}
    \tau_{\min} \approx (1/r) \log(S_0) + \gamma,
\end{equation}
where $\gamma \approx 0.5772$ is the Euler–Mascheroni constant. Note that this is a conservative lower bound of the stopping time since we assume that every division results in a loss of one stem cell (i.e. $p_3(t) = 1 \ \forall t$). Figure \ref{fig:sim_plot} (right) presents an example of stopping time distribution from simulated data. The minimum expected stopping time lies below all simulated stopping times.

We next fit the stopping times from simulated data using the inverse Gaussian distribution. The inverse Gaussian distribution is commonly used to model first passage or stopping times of processes such as Brownian motion with drift \cite{folks1978inverse}. In our context, the inverse Gaussian distribution provides a natural model for the time it takes for the number of viable stem cells to reach 0. The simulated stopping data is obtained from 100 realizations of the process with the initial count of 200, division rate of 0.2, and division probabilities $p_1(t) = \frac{0.55}{1 + 0.005(t-4)^2}$, $p_2(t) = \frac{0.15}{1 + (t-12)^2}$, $p_4(t) = \frac{0.20}{1+(t-20)^2}$, and $p_3(t) = 1-p_1(t)-p_2(t)-p_4(t)$. Since the support of the inverse Gaussian is the positive real line, we shift the stopping times by $-\tau_{\min}$. Using these data, we estimate the mean and shape parameter of the inverse Gaussian distribution by the maximum likelihood method. The estimated parameters are then used to fit the distribution to the simulated data, and the goodness-of-fit is assessed through standard distributional tests. The results of distributional tests in table \ref{tab:stopping_time_dist} and diagnostic plots in figure \ref{fig:stopping_time_dist} in the Supplementary Materials indicate that the inverse Gaussian distribution is a good fit for the shifted simulated stopping times.
\begin{table}[h]
    \centering
    \begin{tabular}{lcc}
        \toprule
        Distribution test & Test statistics & P-value \\
        \midrule
        Kolmogorov-Smirnov & 0.066 & 0.766 \\
        Anderson-Darling & 0.588 & 0.658 \\
        \bottomrule
    \end{tabular}
    \caption{Distributional tests of the shifted stopping times with the inverse Gaussian distribution.}
    \label{tab:stopping_time_dist}
\end{table}
\section{Likelihood Estimation} \label{mle}
In practice, the division probabilities $p_1(t), p_2(t), p_3(t), p_4(t)$ are unknown and must be estimated using the observed cell proliferation data from experiments. In this section, we
 a likelihood-based framework for parameter estimation. When we observe all the events, both the timing of the events and the number of cells, we can write down the full likelihood of the observed data
\begin{equation} \label{likelihood}
    \begin{split}
        & \LL(\boldsymbol{\theta}; \mathbb{T},\mathbb{X},  \mathbb{Y},\mathbb{Z}) = \prod_{i=1}^n P(\Delta X_i, \Delta Y_i, \Delta Z_i \vert T_i) f(\Delta T_i\vert X_{i-1}),
    \end{split}
\end{equation}
where $\boldsymbol{\theta}$ denotes the set of parameters for $p_j(t)$ and division rate $r$, $P(\Delta X_i, \Delta Y_i, \Delta Z_i \vert T_i)$ is the distribution of the change of cell counts in (\ref{marginal}), and $f(\cdot\vert  X_{i-1})$ denotes the density function of the exponential distribution with rate $rX_{i-1}$. Note that in the likelihood for fully observed data, we only need to use the observed viable stem cells and nonviable stem cells to determine which division event occurs. The likelihood can be optimized to obtain the maximum likelihood estimates of $p_j(t)$ parameters and $r$. In practical experiments, we may encounter situations in which viable and nonviable stem cells cannot be distinguished. As a result, instead of observing $X(t)$ and $Z(t)$, we observe the sum of the viable and nonviable stem cells denoted by $M(t)$. As seen in (\ref{likelihood}), the unobserved number of nonviable stem cells plays an important role in exact likelihood evaluation. A na\"ive approach to evaluate the likelihood when the nonviable stem cells are not distinguishable from viable stem cells is to sum out all the trajectories of possible numbers of nonviable stem cells, which is computationally inefficient. In such cases, we employ a modification of the forward algorithm, which is a dynamic programming application developed for hidden Markov models  \cite{rabiner2002tutorial}, to write the likelihood of partially observed data.

We first define the forward variables $\alpha_k(u)$ (for $k = 1, \cdots, n$) as the likelihood that there are $u$ unobserved nonviable stem cells at $k$-th division events, given all observations $(M_i, Y_i, T_i)$ up to and including $k$-th division events
$$\alpha_k(u) = \LL(\mathbb{T}_k,\mathbb{Y}_k, \mathbb{M}_k, Z_k = u, \boldsymbol{\theta}),$$
where $u$ is an integer and $0 \leq u \leq M_k$, $\mathbb{T}_k, \mathbb{Y}_k, \mathbb{M}_k$ are the observations of event times, differentiated cells, and total stem cells up to $k$-th division events, respectively. The forward variables can be computed recursively by
$$\alpha_{k+1} (v) = \sum_{u \in\{v, v-1\}} \alpha_k(u) \cdot h_k(u,v), \quad 0 \leq k \leq n-1$$
where $0 \leq v \leq M_{k+1}$. Note that since we assume $T_0 = 0, M_0 = S_0, Y_0 = 0, Z_0 = 0$, we have $\alpha_0(0) = 1$ and $\alpha_0(u) = 0$ for $u >0$. Since the number of nonviable stem cell either stay the same or increase by 1 after each division, we only need to sum over these two cases. Here, $h_k(u,v)$ are analogous to transitional probabilities from having $u$ unobserved nonviable stem cells at step $k$ to having $v$ unobserved nonviable stem cells at step $k+1$. For $i \leq j \leq 4$
\begin{equation*}
    \begin{split}
        h_k(u,v) = \begin{cases}
            h_{jk}(u,v), & \text{event}\ j \ \text{at} \ T_{k+1}\\
            0, & \text{otherwise},
        \end{cases}
    \end{split}
\end{equation*}
$h_{jk}(u,v) = r(M_k-u) e^{-r(M_k-u) (T_{k+1}-T_k)}\cdot p_j(T_{k+1})$. The events in terms of $(Y,M,Z)$ are:
\begin{itemize}
    \item Event 1: $Y_{k+1} = Y_k, M_{k+1} = M_k +1, v = u$ with probability $p_1(T_{k+1})$,
    \item Event 2: $Y_{k+1} = Y_k +1, M_{k+1} =M_k, v = u$ with probability $p_2(T_{k+1})$,
    \item Event 3: $Y_{k+1} = Y_k +2, M_{k+1} = M_k -1, v = u$ with probability $p_3(T_{k+1})$,
    \item Event 4: $Y_{k+1} = Y_k, M_{k+1} = M_k + 1, v = u+1$ with probability $p_4(T_{k+1})$.
\end{itemize}
The likelihood is given by $$\LL_n = \sum_{0 \leq u \leq M_n} \alpha_n(u).$$
To avoid the issue of underflow or overflow in computation, we apply a standard normalization technique to the forward algorithm. The normalized forward variable is defined as
$$\bar{\alpha}_k(u) = \frac{\alpha_k(u)}{\LL_k},$$
where $\LL_k = \displaystyle\sum_{0 \leq u \leq M_k} \alpha_k (u), \ k=1, \cdots, n$. The recursion of the normalized forward variable is then given by
\begin{equation}
    \begin{split}
        \bar{\alpha}_{k+1}(v) & = \frac{\alpha_{k+1} (v)}{\LL_{k+1}} =\frac{1}{\LL_{k+1}} \sum_{u = v-1,v} \alpha_k(u) h_k(u,v) = \frac{\LL_k}{\LL_{k+1}} \sum_{u = v-1,v} \bar{\alpha}_k h_k(u,v).
    \end{split}
\end{equation}
For $0 \leq k \leq n-1$, we define the sequential likelihood ratio
\begin{equation}
    \begin{split}
        d_{k+1} & = \frac{\LL_{k+1}}{\LL_k} \\
        & = \frac{1}{\LL_k} \sum_{0 \leq v\leq M_{k+1}} \alpha_{k+1}(v) \\
        & = \frac{1}{\LL_k} \sum_{0 \leq v \leq M_{k+1}} \sum_{u = v-1, v} \alpha_k(u) h_k (u,v)\\
        & = \sum_{0\leq v\leq M_{k+1}} \sum_{u=v-1,v} \bar{\alpha}_k(u) h_k(u,v).
    \end{split}
\end{equation}
The log-likelihood of the observed data is obtained by taking the logarithm of the sequential likelihood ratios, and performing the telescopic summation
$$\log(\LL_n)=\sum_{k=0}^{n-1}  \log(d_{k+1}).$$
Algorithm \ref{forward_algorithm} provides a formal description of the normalized forward algorithm. It can be used to optimize and obtain the maximum likelihood estimation. The results of the algorithm implementation and parameter estimation are presented in the next section.

\begin{algorithm}
\caption{Normalized Forward Algorithm to Compute Log-likelihood ($\log \mathbb{L}_n$)}
\begin{algorithmic}[1]
\Require
Parameter of probabilities $p_j(t)$ and division rate $r$, observed event times $\mathbb{T} = (t_1, \cdots, t_n),$ observed total stem cells $\mathbb{M} = (m_1, \cdots m_n)$ and differentiated cells $\mathbb{Y} = (y_1, \cdots, y_n)$.
\Ensure Log-likelihood $\log \mathbb{L}_n$.

\Statex \textbf{Initialization:}
    \State $\bar{\alpha}_0(0) = 1$
\Statex \textbf{Recursion:}
\For{$k = 0$ to $n-1$}
    \For{$u \in \{v-1,v\}$, $0 \leq v \leq M_{k+1}$}
        \State $h_k(u,v)$
    \EndFor

    \State $d_{k+1}
    = \displaystyle\sum_{0\leq v\leq M_{k+1}} \sum_{u=v-1,v} \bar{\alpha}_k(u) h_k(u,v)$

    \For{$ u \in \{v-1, v\}$}
        \State $\bar{\alpha}_{k+1}(v)
        = \frac{1}{d_{k+1}}
        \displaystyle\sum_{u} \bar{\alpha}_k(u) h_k(u,v)$
    \EndFor
\EndFor

\Statex \textbf{Termination:}
\State $\log \mathbb{L}_n
    = \displaystyle\sum_{k=0}^{n-1} \log d_{k+1}$

\State \Return $\log \mathbb{L}_n$

\end{algorithmic}
\label{forward_algorithm}
\end{algorithm}
\section{Simulation Results} \label{simulation}
In this section, we assess the performance of the maximum likelihood estimator (MLE) for the branching process model of the stem cell proliferation process. We consider both fully observed cell count data and scenarios in which only partial observations are available. Simulated cell count trajectories are generated as follows. Starting from a fixed number of viable stem cells, the time until the next event was drawn from an exponential distribution with rate proportional to the current number of viable stem cells. At each event, a stem cell either divides or differentiates according to probabilities $p_j(t)$. In our simulation, the probabilities $p_j(t)$ are time-dependent with the functional form described in (\ref{lorentzian}). After each event, the population of each cell type is updated, and the process is repeated until there are no viable stem cells left. For partially observed data, we do not observe the counts of viable and nonviable stem cells separately, but only observe their combined total. In our estimation method, we only need to track the timing and changes in cell counts of each division event rather than the full lineage of individual cells. As a result, our estimation method reduces the need for extensive experimental data collection. We experimented with three parameter configurations of the division probabilities $p_j(t)$ to assess the performance of the maximum likelihood estimator. Figure \ref{fig:prob} displays the plots of probabilities $p_j(t)$, $j = 1,2,3,4$ with these three different parameter configurations.
Configuration 1 represents an early dominance of the symmetric self-renewal divisions with a large initial peak in $p_1(t)$. Configuration 2 provides a balanced case in which functions $p_1(t), p_2(t), p_4(t)$ have the equal maximum and rate of decay, but the maximums occur at well-separated times. Configuration 3 emphasizes on a strong peak in $p_2(t)$, which is the asymmetric division into one stem cell and one differentiated cell. For each parameter configuration, 100 replicate trajectories are simulated. The number of initial viable stem cells and division rate used in the simulation are 200 and 0.2, respectively.
\begin{figure}
    \centering
    \includegraphics[width=1\linewidth]{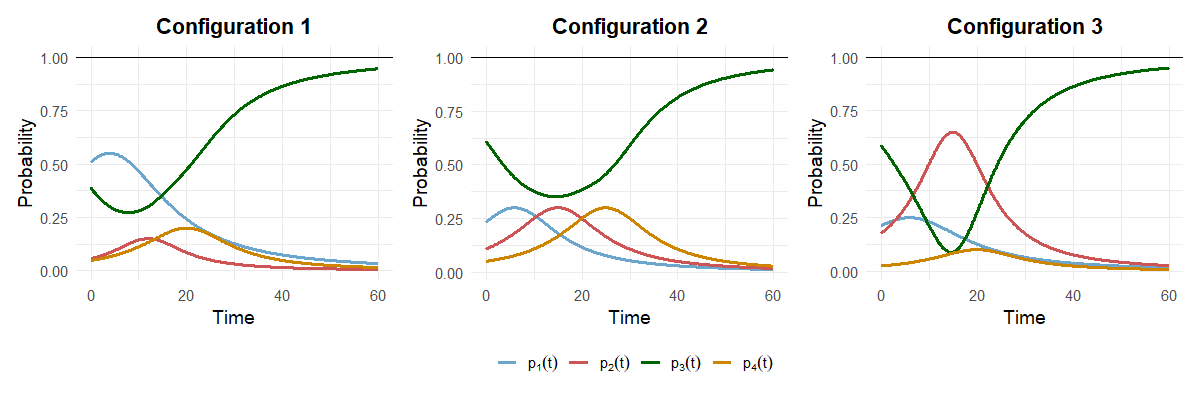}
    \caption{Time-dependent probability functions for division and differentiation events under three Lorentzian parameter configurations. Each panel displays the four probabilities $p_1(t), p_2(t), p_3(t),$ and $p_4(t)$ defined in equation (\ref{lorentzian}).}
    \label{fig:prob}
\end{figure}
The likelihood of fully observed data is computed based on equation (\ref{likelihood}). For numerical stability, we take the logarithm of the likelihood of fully observed data
\begin{equation*}
    \begin{split}
        & \ell(\boldsymbol{\theta},\mathbb{T}, \mathbb{X}, \mathbb{Y}, \mathbb{Z}) \\
        &= \sum_{i=1}^k \Big[ \log \Big(\frac{p_1}{1+c_1(T_i-m_1)^2}\Big) I_{(\Delta X_i = 1, \Delta Z_i = 0)} \\
        & + \log \Big(\frac{p_2}{1+c_2(T_i-m_2)^2}\Big) I_{(\Delta X_i = 0,\Delta Z_i = 0)} \\
        & + \log \Big(\frac{p_4}{1+c_4(T_i-m_4)^2}\Big) I_{(\Delta X_i = 0,\Delta Z_i = 1)} \\
        & + \log \Big(1- \frac{p_1}{1+c_1(T_i-m_1)^2}- \frac{p_2}{1+c_2(T_i-m_2)^2} - \frac{p_4}{1+c_4 (T_i-m_4)^2} \Big) I_{(\Delta X_i = -1,\Delta Z_i = 0)} \\
        & + \log r + \log S_{i-1} - rS_{i-1}\Delta_{T_1} \Big].
    \end{split}
\end{equation*}
The closed-form solution of the MLE for division rate $r$ can be obtained by taking the derivative of the log-likelihood with respect to $r$. The likelihood of the partially observed cell counts is computed using the forward algorithm described in section \ref{mle}. Parameters are estimated by numerical maximization of the likelihood: for fully observed data, we use the differential evolution (DE) algorithm for parameters of $p_j(t)$, while the MLE of division rate can be obtained through its closed-form solution. DE is a global optimization algorithm which maintains a population of candidate solutions and iteratively improves them through mutation, crossover, and selection. In the mutation step, new candidate solutions are generated by adding a scaled difference between two population members to a third member, allowing the algorithm to explore the parameter space efficiently. These candidates are then mixed with existing solutions through crossover to increase diversity of the population of candidates, with the best-performing candidates are selected to retain for the next generation. DE is particularly well-suited for complex or non-convex likelihood surfaces, as it does not require gradient information and is less likely to be trapped in local optima \cite{storn1997differential, ardia2011differential, mullen2011deoptim}. For partially observed data, DE is combined with the gradient-based Broyden-Fletcher-Goldfarb-Shanno (BFGS) optimization, with the gradient computed by automatic differentiation \cite{fournier2012ad, kristensen2016tmb} to refine estimates near local maxima. This hybrid approach takes advantage of the global search capability of DE to explore the parameter space broadly, while BFGS with automatic differentiation refines solutions locally, improving convergence speed and numerical stability. Note that although we employ automatic differentiation to evaluate the gradient of the normalized forward-algorithm-based likelihood, the gradient can also be computed dynamically. The full derivation is provided in the Supplementary Materials.
\begin{table}[h]
    \centering
    \begin{tabular}{lccccc}
        \toprule
         Parameter & True value & Mean & Median  & 5.0 Percentile & 95.0 Percentile \\
         \midrule
        \multicolumn{4}{l}{\textbf{Parameter configuration 1}} \\
         $p_1$ & 0.550 &  0.560 & 0.557 & 0.524 & 0.602 \\
         $p_2$ & 0.150 & 0.151 & 0.148 & 0.123 & 0.181\\
         $p_4$ & 0.200 & 0.201 & 0.200 & 0.175 &     0.242 \\
         $c_1$ & 0.005 & 0.00504 & 0.00491 & 0.00365 & 0.00704\\
        $c_2$ & 0.012 & 0.01270 & 0.01170 & 0.00653 & 0.02050 \\
        $c_4$ & 0.008 & 0.00854 & 0.00826 & 0.00543 & 0.01280\\
        $m_1$ & 4 & 3.500 & 3.740 & 0.549 & 5.750  \\
        $m_2$ & 12 & 11.900 & 11.900 & 10.100 & 13.500  \\
        $m_4$ & 20 & 19.900 & 20.000 & 18.000 & 21.800 \\
        $r$ & 0.200  & 0.199 & 0.199 &  0.191 & 0.208\\
        \multicolumn{4}{l}{\textbf{Parameter configuration 2}} \\
        $p_1$ & 0.300 & 0.306 & 0.303 & 0.256 & 0.373\\
        $p_2$ & 0.300 & 0.304 & 0.303 & 0.250 & 0.354 \\
        $p_4$ & 0.300 & 0.306 & 0.305 & 0.246 & 0.372\\
        $c_1$ & 0.008 & 0.00894 & 0.00790 & 0.00389 & 0.01690\\
        $c_2$ & 0.008 & 0.00872 & 0.00800 & 0.00465 & 0.01390 \\
        $c_4$ & 0.008 & 0.00882 & 0.00831 & 0.00442 & 0.01430\\
        $m_1$ & 6 & 5.650 & 5.920 & 2.21 & 8.29 \\
        $m_2$ & 15 & 15.100 & 15.200 & 13.300 & 17.2 \\
        $m_4$ & 25 & 24.900 & 25.000 & 22.000 & 28.500\\
        $r$ & 0.200 & 0.199 & 0.199 & 0.186 & 0.211\\
        \multicolumn{4}{l}{\textbf{Parameter configuration 3}} \\
        $p_1$ & 0.250 & 0.258 & 0.258 & 0.223 & 0.295\\
        $p_2$ & 0.650 & 0.649 & 0.652 & 0.585 & 0.709  \\
        $p_4$ & 0.100 &  0.105 & 0.102 & 0.0713 & 0.147 \\
        $c_1$ & 0.005 & 0.00560 & 0.00526 & 0.00272 & 0.00938 \\
        $c_2$ & 0.012 & 0.01220 & 0.01220 & 0.00887 & 0.01590 \\
        $c_4$ & 0.008 & 0.01070 & 0.00927 & 0.00326 & 0.02130  \\
        $m_1$ & 6 & 5.590 & 6.130 & 0.501 & 9.290\\
        $m_2$ & 15 & 15.000 & 15.000 & 14.200 & 15.700 \\
        $m_4$ & 20 & 19.500 & 19.500 & 16.000 & 23.400 \\
        $r$ & 0.200 & 0.199 & 0.198 & 0.189 & 0.210\\
        \bottomrule
    \end{tabular}
    \caption{Parameter estimations under fully observed data for the three parameter configurations.}
    \label{tab:par_est_full}
\end{table}

For the fully observed data case (Table \ref{tab:par_est_full}), the means and medians of the estimates closely match the true parameter values and the 5th-95th percentile intervals are narrow across all three configurations. Table \ref{tab:par_est_forward} summarizes the estimates obtained with the normalized forward algorithm for scenarios in which we cannot distinguish viable and nonviable stem cells. The result shows a slight increase in the variability of the estimates as the percentile ranges are slightly wider. The true parameters are still contained within the corresponding intervals for all parameters. Across three configurations, the location parameters $m_i$ are estimated particularly well, whereas the estimates of parameters $c_i$ experience higher degree of sensitivity to loss of information due to unobserved data.
\begin{table}[h]
    \centering
    \begin{tabular}{lccccc}
        \toprule
         Parameter & True value & Mean & Median  & 5.0 Percentile & 95.0 Percentile \\
        \midrule
        \multicolumn{4}{l}{\textbf{Parameter configuration 1}} \\
        $p_1$ & 0.550 & 0.565 & 0.563 &  0.512 & 0.621 \\
        $p_2$ & 0.150 & 0.151 & 0.149 & 0.121 & 0.182 \\
        $p_4$ & 0.200 & 0.210 & 0.205 &  0.170 & 0.261\\
        $c_1$ & 0.005 & 0.00539 & 0.00513 & 0.00286 & 0.00875 \\
        $c_2$ & 0.012 & 0.01280 & 0.01180 & 0.00617 & 0.02070\\
        $c_4$ & 0.008 & 0.01030 & 0.00959 & 0.00447 & 0.01920\\
        $m_1$ & 4 & 3.510 & 3.700 & 0.01850 & 6.070 \\
        $m_2$ & 12 & 11.900 & 12.000 & 10.000 & 13.500 \\
        $m_4$ & 20 & 19.900 & 20.100 & 15.700 & 22.900 \\
        $r$ & 0.200  & 0.198 & 0.197 &  0.187 & 0.213\\
        \multicolumn{4}{l}{\textbf{Parameter configuration 2}} \\
        $p_1$ & 0.300 & 0.313 & 0.315  & 0.246 & 0.375\\
        $p_2$ & 0.300 & 0.304 & 0.305 &
        0.247 & 0.355 \\
        $p_4$ & 0.300 & 0.311  & 0.307  & 0.223 & 0.409\\
        $c_1$ & 0.008 & 0.00979 & 0.00898 & 0.00267 & 0.02000 \\
        $c_2$ & 0.008 & 0.00879 & 0.00808 &
        0.00446 & 0.01460\\
        $c_4$ & 0.008 & 0.01000 & 0.00892 & 0.00351 & 0.01910 \\
        $m_1$ & 6 & 5.610  & 6.010 & 0.184 & 9.030 \\
        $m_2$ & 15 & 15.100   & 15.100   &  13.200 & 17.200 \\
        $m_4$ & 25 & 24.900   & 25.000     & 21.200 & 28.800 \\
        $r$ & 0.200  & 0.197  & 0.196  & 0.183 & 0.218 \\
        \multicolumn{4}{l}{\textbf{Parameter configuration 3}} \\
        $p_1$ & 0.250 & 0.271 & 0.275 & 0.208 & 0.329\\
        $p_2$ & 0.650 & 0.658 & 0.662 & 0.587 & 0.714 \\
        $p_4$ & 0.100 & 0.135 & 0.129 & 0.08750 & 0.207 \\
        $c_1$ & 0.005 & 0.00746 & 0.00711 & 0.00179 & 0.01460 \\
        $c_2$ & 0.012 & 0.01260 & 0.01230 & 0.00906 & 0.01680 \\
        $c_4$ & 0.008 & 0.02440 & 0.01640 & 0.00314 & 0.06360 \\
        $m_1$ & 6 & 6.610 & 6.480 & 0.227 & 15.100 \\
        $m_2$ & 15 & 15.000 & 15.000 & 14.200 & 15.800  \\
        $m_4$ & 20 & 18.500  & 21.000   & 0.294 & 26.800 \\
        $r$ & 0.200  & 0.198   & 0.196 & 0.183 & 0.224 \\
        \bottomrule
    \end{tabular}
    \caption{Parameter estimation when viable and nonviable stem cells cannot be distinguished using normalized forward algorithm and hybrid optimization.}
    \label{tab:par_est_forward}
\end{table}
We assess the fit of the model by comparing the expected cell count trajectories generated from the estimated parameters from the partially observed data and forward algorithm likelihood with the corresponding observed simulated data. We use the ratio between the observed and expected cell counts
$$ \frac{X(t)}{\hat{X}(t)}, \frac{Z(t)}{\hat{Z}(t)}, \frac{M(t)}{\hat{M}(t)}, \frac{Y(t)}{\hat{Y}(t)},$$
goodness-of-fit assessment metrics. If the ratios are close to 1, the model provides a good estimation of the trajectory of cell counts. Figure \ref{fig:prediction_ratio} depicts the observed-to-expected cell count ratios across four panels: viable stem cells, nonviable stem cells, total stem cells, and differentiated cells. The ratios for differentiated cells are tightly clustered around one, indicating a good fit of the model to the observed data. Total stem cell ratios are also well estimated. The ratios for nonviable stem cells show the greatest variability, especially at early time points. Overall, these results suggest that the estimation of the model parameters accurately captures the dynamics of differentiated and total stem cells, with larger discrepancies primarily occurring for nonviable stem cells during the initial phase.
\begin{figure}[h]
    \centering
    \includegraphics[width=1\linewidth]{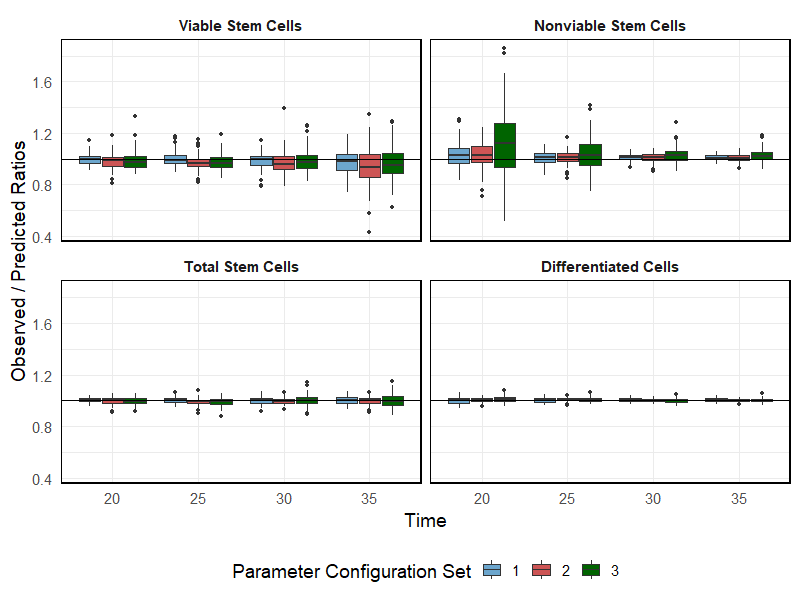}
    \caption{Ratios of observed-to-expected cell counts. The predicted cell counts are generated by using the estimated parameters (in table \ref{tab:par_est_forward}) for partially observed data using forward algorithm likelihood.}
    \label{fig:prediction_ratio}
\end{figure}

\section{Application to Bipotent Megakaryocytic Erythroid Progenitor Clonal Expansion and Differentiation} \label{application}
We apply our branching process model to bipotent megakaryocytic erythroid progenitor (MEP) cell proliferation. MEPs are hemtopoietic progenitors that generate both megakryocytic destined progenitors (MkPs) and erythroid destined progenitors (ErPs), which ultimately produce platelets and red blood cells, respectively. MEPs are not stem cells, however, they display stem cell-like behavior by undergoing self-renewal divisions that maintain bipotency or differentiation events that yield progenitors restricted to a single lineage. The data on MEP clonal expansion and differentiation, collected by Scanlon et al. \cite{Sc2022} using single-cell time-lapse imaging with \textit{in situ} fluorescence staining, indicated that MEPs exhibit division patterns similar to stem cells, including symmetric divisions that produce two bipotent MEPs, asymmetric divisions that generate one bipotent MEP and one megakaryocytic or erythroid progenitor, and differentiative divisions that result in lineage-committed progeny. Furthermore, the probability that MEPs undergo either symmetric or asymmetric self-renewal changes dynamically over time. In the experiment, MEPs were cultured in a medium containing the standard cytokine cocktail: recombinant human Interleukin-3, Interleukin-6, Stem Cell Factor, thrombopoietin (TPO), and erythropoietin (EPO). An additional group of MEPs were also cultured in the same medium without TPO to assess its role in regulating MEPs fate decision.
\begin{figure}[h]
    \centering
    \includegraphics[width=1\linewidth]{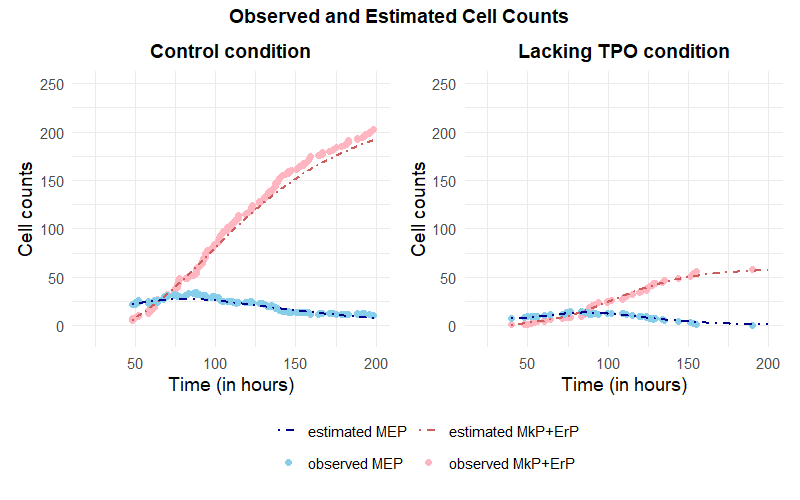}
    \caption{Observed and estimated cell counts calculated from the estimated parameters for the experimental MEP clonal expansion and differentiation data.}
    \label{fig:MEPexperiement}
\end{figure}

The originally published data in Scanlon et al. \cite{Sc2022} recorded each cell's generation, first and last observation times, and division outcomes. This allowed each cell to be followed from its parent to any daughter cells resulting from division. Cell counts were updated based on the observed outcomes, organized according to the time of division. Accurate timing is critical for preserving the time-dependent patterns to estimate the rate of division and the probability of division events. Therefore, divisions occurring very early in the experiment, where exact timing could not be precisely determined, were excluded to ensure reliable lineage tracking. The processed data thus capture the counts of MEPs and the combined MkP and ErP progeny at each division event. Since there are very small numbers of MEPs retained at the end of the process (3 MEPs retained in the data of MEPs cultured in control condition and no MEP retained in the data of MEPs cultured in lacking TPO condition), we apply the model without the nonviable stem cells and estimate the parameters using the fully observed data likelihood estimation.

Figure \ref{fig:MEPexperiement} compares the experimentally observed counts of MEPs and of the total differentiated progeny with the estimated counts computed by the model based on the estimated parameters. We use the functions described in (\ref{lorentzian}) for the division probabilities with the exceptions of $p_4(t) = 0 \ \forall t$. In the TPO-depleted condition, the estimated counts closely match the observed numbers. In the control condition, the model slightly underestimates the number of total differentiated progeny at later time points, but the overall trends are still well-represented. The estimated division rate of MEPs in the control condition and lacking TPO condition is 0.059 and 0.045, respectively. Consistent with the result in Scanlon et al. \cite{Sc2022}, the decreased division rate observed in the MEPs cultured in the TPO-deplete condition reflects an increase in the time between consecutive division compared with those cultured in control condition.

\section{Discussion} \label{discussion}
In this work, we have formulated a continuous-time branching process model for cell proliferation with time-dependent division probabilities. Unlike classical branching models with fixed offspring distribution, our framework accommodates proliferation processes in which the cell fate decisions may evolve during development.

One key contribution in this work is the development of the likelihood-based inference procedures under two different observation scenarios. When all division events and cell counts are observable, we derived a full likelihood that allows direct parameter estimation. More importantly, when cell counts are partially unobservable - viable and non-viable stem cells cannot be distinguished - we constructed a forward algorithm likelihood. This approach avoids the need to calculate the likelihood from all possible divisions resulting in the observed cell counts, greatly reducing the complexity of the likelihood computation. Simulation studies demonstrate that the proposed estimation recovers the time-dependent division probabilities with good accuracy under both scenarios from the aggregate cell count data and the timing of the division events. In both scenarios, our method does not require to track individual cell lineage, which are often difficult to obtain in experiments. Our results suggest that meaningful information about proliferation probabilities can still be recovered where lineage tracking is not feasible. From an experimental perspective, this substantially reduces the work required to track and collect data.

Our proposed framework can be extended in useful ways to study more complex cell proliferation processes. The number of division outcomes to can be expanded to include additional progenitor or differentiated cell types, leading to higher-dimensional multi-type branching processes. In addition, while the current model allows division probabilities to be time-dependent, they can be generalized to depend on other covariates, such as spatial location or local cell density. We expect that our approach will be useful for analyzing a wide range of biological processes in which cell fate decisions evolve dynamically and are only partially observable.

\section*{Data Availability}
The simulation code, processed data and analysis code for the experiment data example are available at \url{https://github.com/huyendn/branching_process_cell}. The original dataset used in the example is publicly available and was obtained from Scanlon et al. \cite{Sc2022}, where details on access are provided.

\section*{Supplementary Materials} \label{supplement}
\subsection*{Fitting Inverse Gaussian Distribution to the Stopping Time}
The simulated stopping data is obtained from 100 realizations of the process with the initial count of 200, division rate of 0.2, and division probabilities $p_1(t) = \frac{0.55}{1 + 0.005(t-4)^2}$, $p_2(t) = \frac{0.15}{1 + (t-12)^2}$, $p_4(t) = \frac{0.20}{1+(t-20)^2}$, and $p_3(t) = 1-p_1(t)-p_2(t)-p_4(t)$. The simulated stopping time data is then shifted by $-\tau_{\min}$. The maximum likelihood estimates of mean and shape parameter of the inverse Gaussian distribution are shown in table \ref{tab:stopping_time_est}. The estimates are then used to fit the distribution to the simulated data.
\begin{table}[h]
    \centering
    \begin{tabular}{lcc}
        \toprule
        Parameter & Mean & Shape \\
        Estimates & 36.858 & 586.698 \\
        \bottomrule
    \end{tabular}
    \caption{Maximum likelihood estimates of inverse Gaussian distribution parameters with the shifted simulated stopping time.}
    \label{tab:stopping_time_est}
\end{table}
\begin{figure}
    \centering
    \includegraphics[width=1\linewidth]{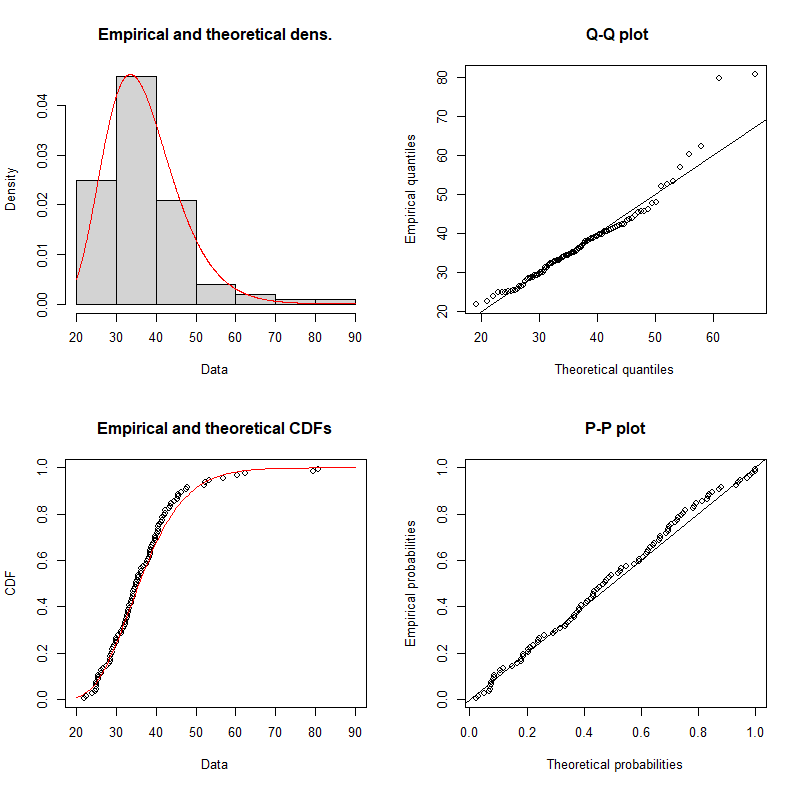}
    \caption{Diagnostic plots assessing the fit of the inverse Gaussian distribution to the stimulated stopping times. The fitted distribution is based on parameters estimated using maximum likelihood.}
    \label{fig:stopping_time_dist}
\end{figure}
\subsection*{Gradient of the Forward Likelihood}
For $1 \leq v \leq M_{k+1}$, the recursion of the normalized forward variable is given by
$$\bar{\alpha}_{k+1}(v)  = \frac{\LL_k}{\LL_{k+1}} \sum_{u = v-1,v} \bar{\alpha}_k h_k(u,v) = \frac{1}{d_{k+1}} \sum_{u = v-1,v} \bar{\alpha}_k h_k(u,v).$$
Since $$d_{k+1} = \displaystyle \sum_{0 \leq v \leq M_{k+1}} \displaystyle\sum_{u = v-1,v} \bar{\alpha}_k h_k(u,v),$$ we have the gradient using the chain rule $$ d'_{k+1} = \displaystyle \sum_{0 \leq v \leq M_{k+1}} \displaystyle\sum_{u = v-1,v} \Big[ \bar{\alpha}'_k h_k(u,v) + \bar{\alpha}_k h'_k(u,v)\Big].$$
Using the chain rule again, the recursion of the gradient of the normalized forward variable is given by
\begin{equation}
    \begin{split}
        \bar{\alpha}_{k+1}(v) = \frac{d'_{k+1}}{(d_{k+1})^2} \sum_{u = v-1, v} \bar{\alpha}_k h_k(u,v) + \frac{1}{d_{k+1}} \sum_{u = v-1, v} \Big[ \bar{\alpha}'_k h_k(u,v) + \bar{\alpha}_k h'_k(u,v) \Big].
    \end{split}
\end{equation}
The gradient of the log-likelihood is then
$$ (\log \LL_n)' = \sum_{k=1}^{n-1} \frac{d'_{k+1}}{d_{k+1}}.$$
\subsection*{R packages}
All analyses, simulations, and computational implementations in this study were performed using R version 4.4.2.\cite{R-base}. Distribution fitting and goodness-of-fit diagnostics were performed using packages \textbf{fitdistrplus, goftest,} and \textbf{stats} \cite{fitdistrplus2015, goftest2021, R-base}. Likelihood-based parameter estimation via differential evolution algorithm, gradient-based optimization, and automatic differentiation were carried out using \textbf{DEoptim}, \textbf{parallelly}, \textbf{base R}, and \textbf{TMB} \cite{mullen2011deoptim,  bengtsson2025parallelly, R-base,kristensen2016tmb}. Figures were generated using \textbf{ggplot2} \cite{wickham2016ggplot2} and base R \cite{R-base}.
\bibliography{cell_arxiv}

\end{document}